\documentclass[conference]{IEEEtran}
\IEEEoverridecommandlockouts

\usepackage{cite}
\usepackage{amsmath,amssymb,amsfonts}
\usepackage{algorithmic}
\usepackage[ruled,lined]{algorithm2e}
\usepackage{graphicx}
\usepackage{textcomp}
\usepackage{xcolor}
\usepackage{mathtools}
\usepackage{siunitx}
\usepackage{bm}
\usepackage{subfigure}
\usepackage{float}

\DeclareMathOperator{\Tr}{\mathrm{Tr}}
\newtheorem{theorem}{Theorem}
\newtheorem{definition}{Definition}
\newtheorem{lemma}[theorem]{Lemma}

\newtheorem{condition}[definition]{Condition}

\newcommand{\argmin}{\mathop{\rm argmin}\limits}

\newcommand{\nc}{\newcommand}
\nc{\ketbra}[2]{|#1\rangle\!\langle#2|}
\nc{\bra}[1]{\langle#1|}
\nc{\ket}[1]{|#1\rangle}
\nc{\st}{{\text{ s.t. }}}
\nc{\tr}{\operatorname{Tr}}

\def\BibTeX{{\rm B\kern-.05em{\sc i\kern-.025em b}\kern-.08em
    T\kern-.1667em\lower.7ex\hbox{E}\kern-.125emX}}

\nc{\cM}{{\cal M}}
\nc{\cN}{{\cal N}}
\nc{\cS}{{\cal S}}
\nc{\cH}{{\cal H}}
\nc{\cG}{{\cal G}}
\nc{\cF}{{\cal F}}
\nc{\cI}{{\cal I}}
\nc{\cD}{{\cal D}}

\def\Label#1{\label{#1}\ [\ \text{#1}\ ]\ }
\def\Label{\label}

\begin{document}

\title{A Posteriori Certification Framework for Generalized Quantum Arimoto–Blahut Algorithms\thanks{M.H. was supported in part by 
the Guangdong Provincial Quantum Science Strategic Initiative (Grant No. GDZX2505003), 
the General R\&D Projects of 1+1+1 CUHK-CUHK(SZ)-GDST Joint Collaboration Fund (Grant No. GRDP2025-022), 
and
the Shenzhen International Quantum Academy (Grant No. SIQA2025KFKT07).}}

\author{
\IEEEauthorblockN{Geng Liu}
\IEEEauthorblockA{\textit{School of Science and Engineering} \\
\textit{The Chinese University of Hong Kong, Shenzhen}, China \\
\textit{International Quantum Academy (SIQA)}, Shenzhen, China}
\and
\IEEEauthorblockN{Masahito Hayashi}
\IEEEauthorblockA{\textit{School of Data Science} \\
\textit{The Chinese University of Hong Kong, Shenzhen}, China \\
\textit{International Quantum Academy (SIQA)}, Shenzhen, China \\
\textit{Graduate School of Mathematics, Nagoya University}, Japan \\
email: hmasahito@cuhk.edu.cn}
}

\maketitle

\begin{abstract}
The generalized quantum Arimoto--Blahut (QAB) algorithm is a powerful derivative-free iterative method in quantum information theory. A key obstacle to its broader use is that existing convergence guarantees typically rely on analytical conditions that are either overly restrictive or difficult to verify for concrete problems. We address this issue by introducing an \emph{a posteriori} certification viewpoint: instead of requiring fully \emph{a priori} verifiable assumptions, we provide convergence and error guarantees that can be validated directly from the iterates produced by the algorithm. Specifically, we prove a generalized global convergence theorem showing that, under convexity and a substantially weaker \emph{numerically verifiable} condition, the QAB iteration converges to the global minimizer. This theorem yields a practical certification procedure: by checking explicit inequalities along the computed trajectory, one can certify global optimality and bound the suboptimality of the obtained value. As an application, we develop a certified iterative scheme for computing the quantum relative entropy of channels, a fundamental measure of distinguishability in quantum dynamics. This quantity is notoriously challenging to evaluate numerically: gradient-based methods are impeded by the complexity of matrix functions such as square roots and logarithms, while recent semidefinite programming approaches can become computationally and memory intensive at high precision. Our method avoids these bottlenecks by combining the QAB iteration with \emph{a posteriori} certification, yielding an efficient and scalable algorithm. Numerical experiments demonstrate rapid convergence and improved scalability and adaptivity compared with SDP-based approaches.
\end{abstract}


\section{Introduction}
The Arimoto--Blahut (AB) algorithm is a cornerstone of classical information theory \cite{Arimoto,Blahut},
providing efficient iterative procedures for computing channel capacities, rate--distortion functions, and related quantities.
A key attraction of AB-type methods is that they produce simple closed-form updates and avoid explicit gradient evaluations.
Motivated by the growing need for scalable numerical tools in quantum information theory,
many works have extended AB-type alternating-optimization ideas to quantum settings
\cite{Nagaoka,Dupuis,Sutter,Li-Cai,RISB}.
These quantum/generalized AB algorithms have enabled numerical evaluation of several entropic variational formulas,
including quantum rate--distortion, quantum channel capacities, and quantum information bottleneck problems
\cite{RISB,hayashi2024ab,Hay25}.

Despite their empirical success, AB-type algorithms often face a fundamental theoretical bottleneck:
\emph{global} convergence guarantees typically rely on structural conditions that are difficult to verify for concrete objectives.
In particular, in the generalized quantum AB framework of \cite{hayashi2024ab,Hay25},
one can ensure monotone improvement under a certain inequality condition along updates,
but guaranteeing convergence to the \emph{global} optimum has required additional assumptions
whose verification is frequently intractable in applications.
This mismatch between (i) the practicality of AB iterations and (ii) the difficulty of verifying global-optimality conditions
has limited the scope of AB-type methods as \emph{certified} numerical algorithms.

At the same time, it is important to distinguish our goal from existing \emph{a posteriori stopping} rules.
For example, Ramakrishnan \emph{et al.} \cite{RISB} provide quantum BA algorithms for several capacity formulas,
prove \emph{a priori} $\varepsilon$-approximation guarantees after a prescribed number of iterations,
and propose an \emph{a posteriori stopping criterion} to terminate earlier in practice.
In contrast, the present work focuses on a different and complementary objective:
\emph{a posteriori certification of global optimality and numerical precision} 
\cite{BoydVandenberghe2004,NocedalWright2006,AinsworthOden2000,Verfurth2013}
for generalized quantum AB iterations
in settings where the usual analytic conditions are hard to check.

We introduce an \emph{a posteriori certification} viewpoint for AB-type algorithms.
We prove a generalized global convergence theorem showing that,
for convex objectives, global optimality of the AB fixed point follows from convexity together with a
substantially weaker condition that is \emph{numerically verifiable from the produced iterates}.
This result yields a practical certification procedure:
by checking explicit inequalities along the computed trajectory, one can
(i) certify that the limit point is a global minimizer and
(ii) obtain an explicit bound on the remaining suboptimality.

To demonstrate the impact of our certification framework,
we apply the generalized quantum AB algorithm to compute the quantum relative entropy of channels,
a fundamental measure of distinguishability at the level of quantum dynamics.
For channels $\cN_{A\to B}$ and $\cM_{A\to B}$, it is defined as
\begin{equation}\label{def of qre}
D(\cN\|\cM):=\sup_{\rho_{AR}} D\!\left(\cN_{A\to B}(\rho_{AR})\big\|\cM_{A\to B}(\rho_{AR})\right),
\end{equation}
where the optimization is over all input states with an arbitrary reference system.
Although this quantity has operational interpretations in channel discrimination and resource theories
\cite{cooney2016strong,wilde2020amortized,wang2019resource},
its numerical evaluation is challenging:
gradient-based solvers are hindered by matrix functions such as square roots and logarithms,
while recent SDP-based approaches \cite{wilde2025sdp} can become computationally and memory intensive at high precision.
In this setting, analytically verifying previously known AB convergence conditions is particularly difficult,
whereas convexity of the objective is available \cite{khatri2020principles}.
Our certification theorem bridges this gap: it enables a certified AB computation of $D(\cN\|\cM)$
by verifying the required condition directly from the iterates, yielding an efficient and scalable alternative to SDP.

The remainder of this paper is organized as follows.
Section~\ref{sec:algorithm} reviews the generalized quantum AB iteration and presents our certification criteria.
Section~\ref{sec: application} applies the method to channel relative entropy.
Section~\ref{sec:experiment} reports numerical experiments and compares with the SDP approach in \cite{wilde2025sdp}.
Section~\ref{sec:discussion} concludes with discussion and open directions.

\section{Generalized quantum Arimoto-Blahut algorithm}\label{sec:algorithm}
\subsection{Introduction to general formulation of Arimoto-Blahut Algorithm}\label{subsec:ab algorithm}

The Arimoto-Blahut algorithm is a well-known alternating optimization algorithm to compute the  
classical channel capacities and rate-distortion functions. Recently, many works have extended this algorithm to the quantum setting, and it has been shown that this algorithm has broad applications in various tasks of quantum information theory. Here, we consider the general form of the minimization problem in finite-size quantum systems. 

Firstly, we introduce the mixture family in quantum setting\cite{Amari-Nagaoka,H23}.
For $k$ linearly independent
Hermitian matrices $H_1, \ldots, H_k$ on ${\cal H}$
and a constant vector $c=(c_1, \ldots, c_k)^{T}\in \mathbb{R}^k$, 
we say that a subset $\cM\subset \cS(\cH)$ is the mixture family generated by the $k$ linear constraints
$    \tr\rho H_j = c_j$
for $j=1,...,k$ where ${\cal S}({\cal H})$ is
the set of density matrices on ${\cal H}$.
Then, the mixture family $\cM$ is written as
\begin{align}
{\cal M}:= \{\rho \in {\cal S}({\cal H})| \Tr \rho H_j=c_j \hbox{ for } j=1, \ldots, k
\}.\Label{MDP}
\end{align}
Assuming that $\dim {\cal H}=d$, it is natural to find additional $d^2-1-k$ linearly independent Hermitian matrices 
$H_{k+1}, \ldots H_{d^2-1}$ such that
$H_{1}, \ldots H_{d^2-1}$ are linearly independent.
Thus, we can make parameterization of a density matrix $\rho$ by using the parameter $\eta=(\eta_1, \ldots, \eta_{d^2-1})$ where
$ \eta_j= \Tr \rho H_j$.
In the following,
we expect that the number $k$ of linear constraints is much smaller than $d^2$.


Given a continuous function $\Omega$ from a mixture family ${\cal M}$ to $B({\cal H})$, which is
the set of Hermitian matrices on ${\cal H}$, 
we consider the minimization of ${\cal G}(\rho):= \Tr \rho \Omega[\rho]$ over density matrix $\rho$. 
With the above discussion, our problem is formulated as following minimization problem
\begin{align}
\overline{{\cal G}}(c):=\min_{\rho \in {\cal M}} {\cal G}(\rho)
\label{RBE}
\end{align}
and we denote the minimizer of $\cG(\rho)$ as $\rho_*$.

We then define the $e$-projection of $\rho$ to ${\cal M}$
by 
$\Gamma^{(e)}_{{\cal M}}[\rho]$  \cite{Amari-Nagaoka}, which is defined as
$\Gamma^{(e)}_{{\cal M}}[\rho]:=
\argmin_{ \sigma \in {\cal M}}
D(\sigma\|\rho)$,
where $D(\sigma\|\rho):=\Tr \sigma (\log \sigma-\log \rho)$.
For an  element of $\sigma \in {\cal M}$, the $e$-projection $\Gamma^{(e)}_{{\cal M}}[\rho]$ satisfies Pythagorean theorem \cite[Lemma 5]{H23}:
$D(\sigma\|\rho)=
D(\sigma\|\Gamma^{(e)}_{{\cal M}}[\rho])
+ D(\Gamma^{(e)}_{{\cal M}}[\rho]\|\rho)$.
To calculate the $e$-projection $\Gamma^{(e)}_{{\cal M}}[\rho]$, we need to solve the following equations
\begin{align}
\frac{\partial}{\partial \tau^i}
\log \Tr \exp
(\log \rho + \sum_{j=1}^k H_j \tau^j )
=c_i.
\end{align}
Suppose the solution of the above equations are $\tau_*=(\tau_*^1, \ldots, \tau_*^k)$, the $e$-projection $\Gamma^{(e)}_{{\cal M}}[\rho]$ is given by
$    \Gamma^{(e)}_{{\cal M}}[\rho] = C \exp (\log \rho + \sum_{j=1}^k H_j \tau_*^j )$,
where $C$ is a normalizing constant\cite{H23}.
The solution of the above equations is given as the solution $\tau_*$ of the following minimization problem
\begin{align}
\tau_*:=
\argmin_{\tau }
\log \Big(\Tr \Big(\exp
(\log \rho + \sum_{j=1}^k H_j \tau^j )\Big)\Big)
-\sum_{i=1}^k \tau^i c_i.\Label{BNB}
\end{align}
We note here that the above  function 
is convex for $\tau$ \cite[Section III-C]{H23}.
Since we expect that the number $k$ of linear constraints is much smaller than $d^2$,
this convex minimization is much easier than the original minimization even though 
the original minimization is a convex minimization.

The quantum Arimoto-Blahut algorithm is shown in Algorithm \ref{AL1}, which employs
the conversion function $\cF_3[\cdot]$ from $\cS(\cH)$ to $\cS(\cH)$ given as
${\cal F}_3[\sigma]:= \frac{1}{\kappa[\sigma]}
\exp( \log \sigma -\frac{1}{\gamma} \Omega[\sigma])$,
where $\kappa[\sigma]$ is the normalization factor of
$\Tr \exp( \log \sigma -\frac{1}{\gamma} \Omega[\sigma])$.
When the calculation of $ \Omega[\rho]$ and 
the projection is feasible, Algorithm \ref{AL1} is feasible.

\begin{algorithm}
\caption{Quantum AB algorithm for $\cG(\rho)$}
\Label{AL1}
\begin{algorithmic}
\STATE {Choose the initial value $\rho^{(1)} \in \mathcal{M}$;} 
\REPEAT 
\STATE Calculate $\rho^{(t+1)}:=\Gamma^{(e)}_{{\cal M}}[{\cal F}_3[\rho^{(t)}]]
$;
\UNTIL{convergence.} 
\end{algorithmic}
\end{algorithm}

\subsection{Application to convex minimization}\label{subsec:numerical certified algorithm}
Quantum Arimoto-Blahut algorithm not only has a closed-form update that is physically interpretable and numerically well- behaved, it also enjoys powerful convergence guarantees and can be shown that each iteration always improves the value of the objective function with certain condition. 
Here, we discuss how to apply this algorithm to convex minimization.
\begin{condition}\label{NLT}
    All pairs $(\rho^{(t+1)},\rho^{(t)})$ satisfy 
the following condition with $(\rho,\sigma)=(\rho^{(t+1)},\rho^{(t)})$
\begin{align}
D_\Omega(\rho\|\sigma):=\Tr \rho (\Omega[\rho]- \Omega[\sigma])
&\le \gamma D(\rho\|\sigma)
\Label{BK1+} ,
\end{align}
for a sufficient large number $\gamma$.
\end{condition}

It has been proved in \cite[Theorem 1]{hayashi2024ab}, that if the condition~\ref{NLT} holds for Algorithm~\ref{AL1}, then the algorithm will always iteratively improve the value of the objective function. Thus, this property guarantee the convergence of the Algorithm~\ref{AL1}. However, there is a possibility that the convergent of the algorithm is local minimum instead of the global minimum. To this end, we need following conditions to ensure convergence to the global optimum.


\begin{condition}\Label{con3}
The minimizer $\rho_* $ satisfies
$D_\Omega(  \rho_*\|\sigma) \ge 0$ 
with any state $\sigma\in \cM $.
\end{condition}

According to \cite[Theorem 2]{hayashi2024ab}, we know that if conditions \ref{NLT} and \ref{con3} hold for Algorithm~\ref{AL1}, the algorithm will surely converge to the global minimum $\rho_*$. Moreover, we have
\begin{align}
{\cal G}(\rho^{(t_0+1)})
-{\cal G}(\rho_{*})
\le 
\frac{\gamma D(\rho_*\| \rho^{(1)}) }{t_0} \Label{XME}
\end{align}
 with any initial state $\rho^{(1)}$ and any iteration number $t_0$. Thus, we can estimate the convergence speed by using \eqref{XME}. 

Conditions \ref{NLT} and \ref{con3} do not necessarily hold in general. 
The primary objective of this paper is to provide a guarantee for the property \eqref{XME} in Algorithm~\ref{AL1} even in such cases. 
By relaxing condition \ref{NLT}, we obtain the following theorem, which ensures that Algorithm \ref{AL1} achieves the global minimum under the convexity condition and other weaker requirements.

\begin{theorem}\label{thmA}
We assume the following conditions:
States on ${\cal M}$ are parameterized as $\rho(\theta)$ with $\theta \in \mathbb{R}^d$.
Each matrix component of $\Omega[\rho]$ is a differentiable function of the components of $\rho$.
The value $\Tr \rho(\theta) \Omega[\rho(\theta)]$ is a convex function of $\theta$.
In addition, we assume the following condition:
\begin{description}
\item[(a1)]
The relation \eqref{BK1+} holds 
when $\rho = \rho(\theta_0)$ and 
$\sigma$ is an arbitrary element of the neighborhood of $\rho(\theta_0)$.
\end{description}
Then, the following three conditions are equivalent for $\theta_0$.
\begin{description}
\item[(b1)]
The point $\theta_0$ realizes the minimum value
$\min_{\rho \in {\cal M}}
\Tr \rho  \Omega[\rho]$.
\item[(b2)]
The relation $\Tr \Big(\frac{\partial }{\partial \theta^j}\rho(\theta)\Big|_{\theta=\theta_0} 
\Omega[\rho(\theta_0)]\Big) 
=0$ holds.
\item[(b3)]
$\rho(\theta_0)$ is a convergent of Algorithm \ref{AL1}.
\end{description}
\end{theorem}

While Theorem \ref{thmA} considers convexity with respect to an arbitrary parametrization $\theta$, extending beyond simple convex combinations of density matrices, this paper focuses on cases involving simple convex combinations for our examples. 
The following theorem guarantees the bound \eqref{XME} under weaker conditions.

\begin{theorem}\label{thmB}
Assume the following conditions for the minimizer $\rho_*$ and Algorithm \ref{AL1}:
\begin{description}
\item[(a2)]
The relation 
$D_\Omega(  \rho_*\|\rho^{(j)}) \ge 0 $
holds for $j=1, \ldots, t-1$.
    \item[(a3)]
 The relation 
$    D_\Omega(\rho^{(j+1)}\|\rho^{(j)}) \le \gamma D(\rho^{(j+1)}\|\rho^{(j)}) $
holds for $j=1, \ldots, t-1$.
\end{description}
Then, the evaluation \eqref{XME} holds.
\end{theorem}

These theorems bring us the precision evaluation \eqref{XME}
by a posteriori verification in the following way.
Generally, the conditions in Theorem \ref{thmA} can be verified analytically, except for 
Condition (a1).
To verify this, after running Algorithm \ref{AL1} for a sufficiently large number of iterations $t$, we numerically check whether the relation \eqref{BK1+} holds for $\rho = \rho(\theta_0)$. 
If the ratio $\frac{D_\Omega(\rho^{(t)}\|\sigma)}{D(\rho^{(t)}\|\sigma)}$ is sufficiently smaller than $\gamma$ for a large number of samples $\sigma$
in the neighborhood of $\rho^{(t)}$, we can infer that the ratio for $\rho(\theta_0)$ is also upper-bounded by $\gamma$, because $\rho^{(t)}$ is close to $\rho(\theta_0)$. 
That is, we can conclude that Condition (a1) holds.

To ensure precision, we employ Theorem \ref{thmB}, which requires two conditions. 
If the ratio $\frac{D_\Omega(\rho^{(t)}\|\rho^{(j)})}{D(\rho^{(t)}\|\rho^{(j)})}$ is sufficiently larger than $0$ for $j=1, \ldots,t-1$, 
we can consider Condition (a2) to be satisfied because the convergent point $\rho^{(t)}$ is close to the minimizer $\rho_*$ due to the above discussion. 
Condition (a3) is easier to verify because it only needs to be checked for the $t$ iterates $\rho^{(1)}, \ldots, \rho^{(t)}$ generated by the algorithm. 
Once the above verification is passed, we can evaluate the precision of the obtained minimum value using \eqref{XME}.

\section{Application to calculate the quantum relative entropy of channels}\label{sec: application}

Due to the fundamental roles of quantum channels and the quantum relative entropy of states in quantum information theory, it is natural to ask whether there is a measure that can characterize the distinguishability between two channels. Motivated by this, the quantum relative entropy of channels is developed in \cite{gour2021entropy}. For quantum channels $\cN_{A\to B}$ and $\cM_{A\to B}$, where system $A$ and system $B$ are finite-dimensional Hilbert space, the quantum relative entropy of channels between these two channels is defined as
\begin{equation}
D(\cN\|\cM):= \sup_{\rho_{AR}} D(\cN_{A\xrightarrow{}B}(\rho_{AR})\|\cM_{A\xrightarrow{}B}(\rho_{AR})).\label{eq: def1 of qre of channel}
\end{equation}
The optimization is performed over all possible quantum states $\rho_{AR}$ on the composite space $\cH_{AR}$, where the reference system $R$ could have arbitrary size. Since the dimension of the system $R$ is also involved in this optimization, it is quite hard to compute the channel relative entropy. However, it can be proved that it suffices to optimize over state $\rho_{AR}$ such that $\rho_{AR}$ is pure and the space $R$ is isomorphic to system $A$. This reduction could significantly reduces the complexity of this optimization.

In order to apply the generalized quantum Arimoto-Blahut algorithm to compute the relative entropy of the channel, we need to transform the original optimization problem to a minimization problem whose objective function $\cG(\rho)$ has the form ${\cal G}(\rho)= \Tr \rho \Omega[\rho]$. Use of the Choi matrices of channels $\cN$ and $\cM$ implies the equivalent form of the definition \eqref{eq: def1 of qre of channel}:
\begin{align}\label{concave optimization of qre}
    D(\cN\|\cM) = 
    \sup_{\rho_A \in \cS(\cH_A)} 
    D(\sqrt{\rho_A} \Gamma_{AB}^{\cN}\sqrt{\rho_A}\|\sqrt{\rho_A}\Gamma_{AB}^{\cM}\sqrt{\rho_A}),
\end{align}
where $\Gamma_{AB}^{\cN}$ and $\Gamma_{AB}^{\cM}$ are Choi matrices of quantum channel $\cN$ and $\cM$ respectively. We note that optimization problem in \eqref{concave optimization of qre} is concave in $\rho_A$ \cite{khatri2020principles}. Then $    -D(\cN\|\cM)$ is written as the following convex optimization problem
\begin{align}\label{SJA}
     \inf_{\rho_A \in \cS(\cH_A)} 
\cG(\rho_{A})
\end{align}
where
$\cG(\rho_{A}):=
    -D(\sqrt{\rho_A} \Gamma_{AB}^{\cN}\sqrt{\rho_A}\|\sqrt{\rho_A}\Gamma_{AB}^{\cM}\sqrt{\rho_A}) $, which is known to be convex in $\rho_A$ \cite{khatri2020principles}.
Defining
\begin{align}  &\Omega_1(\rho_{A})\notag\\ &:= -\tr_B(\Gamma_{AB}^{\cN}\rho_A^{\frac{1}{2}} (\log \rho_A^{\frac{1}{2}} \Gamma_{AB}^{\cN}\rho_A^{\frac{1}{2}} - \log \rho_A^{\frac{1}{2}} \Gamma_{AB}^{\cM}\rho_A^{\frac{1}{2}})\rho_A^{-\frac{1}{2}})
\notag\\
     &\Omega(\rho_{A}) :=\frac{\Omega_1(\rho_{A})+\Omega_1(\rho_{A})^{\dagger} }{2},
\end{align}
we can show 
\begin{align}
\cG(\rho_{A})
= \tr[\rho_{A}\Omega_1(\rho_{A})]. \label{ALA1}
\end{align}

Since $\cG(\rho_{A})$ is convex in $\rho_A$,
with a suitable real number $\gamma>0$,
Theorem \ref{thmA} guarantees that the convergent of Algorithm \ref{AL1} is the unique minimizer of \eqref{SJA}.
The matrix ${\cal F}_3[\rho]$ is calculated as
\begin{align}
    \cF_3[\rho_{A}] =& \kappa \exp(\log \rho_{A} - \frac{1}{\gamma}\Omega(\rho_{A})) 
\end{align}
where $\kappa$ is the normalizing constant. 
Since \eqref{SJA} has no linear constraint, 
we do not need to consider $\Gamma^{(e)}_{\cS(\cH_A)}$.
According to the updating rule in Algorithm \ref{AL1}, 
we get $\rho_{A}\in\cS(\cH_A)\mapsto 
{\cal F}_3[\rho_{A}]$ in each iteration. 
Then the update rule is 
\begin{align}
    \rho^{(t+1)} = {\cal F}_3[\rho^{(t)}] =  \frac{\exp(\log \rho^{(t)} - \frac{1}{\gamma}\Omega(\rho^{(t)})) }{\tr(\exp(\log \rho^{(t)} - \frac{1}{\gamma}\Omega(\rho^{(t)})) )}.
\end{align}

When Algorithm \ref{AL1} is applied to the case even with an energy constraint,
we choose a Hamiltonian $H$ on the system ${\cal H}_A$,
and focus on the maximization
\begin{align}
&D_{H,E}(\cN\|\cM)\notag\\
:=& \sup_{\rho_{AR}\in \cS_{(H_j,E_j),A,R}}
D(\cN_{A\xrightarrow{}B}(\rho_{AR})\|\cM_{A\xrightarrow{}B}(\rho_{AR})),
\label{eq: def1 of qre of channelB}
\end{align}
where
\begin{align*}
&\cS_{(H_j,E_j),A,R}\notag\\
:=&\{
\rho_{AR}\in \cS(\cH_A\otimes\cH_R):
\Tr \rho_{AR} H_j\otimes I=E_j , j=1, \ldots,k \}.
\end{align*}
Repeating the same calculaiton, we find that
the problem in \eqref{eq: def1 of qre of channelB} is reduced as
follows
\begin{align}
D_{H,E}(\cN\|\cM)=
\inf_{\rho_{A}\in\cS_{(H_j,E_j)}} \cG(\rho_{A}),\label{SJAB}
\end{align}
where
\begin{align*}
\cS_{(H_j,E_j)}
:=\{\rho_{A}\in\cS(\cH_A):
\Tr \rho_{A}H_j =E_j,j=1,\ldots,k\} .
\end{align*}

Since $\cG(\rho_{A})$ is convex in $\rho_A$,
with a suitable real number $\gamma>0$,
Theorem \ref{thmA} still guarantees that the convergent of Algorithm \ref{AL1} is the unique minimizer of \eqref{SJAB}.
Since we have linear constraints $\Tr \rho_{A}H_j =E_j$,
we need to apply the projection $\Gamma^{(e)}_{\cS_{(H_j,E_j)}}$
after the application of ${\cal F}_3$.
The matrix 
$\Gamma^{(e)}_{\cS_{(H_j,E_j)}}[\cF_3[\rho_{A}]]$
is given as
\begin{align*}
\Gamma^{(e)}_{\cS_{(H_j,E_j)}}[\cF_3[\rho_{A}]] 
=& \kappa \exp(\log \rho_{A} - \frac{1}{\gamma}\Omega(\rho_{A})+\sum_{j=1}^k\tau_*^j H_j), 
\end{align*}
where
$\tau_*=(\tau_*^j) $ is given as
\begin{align}
\tau_*:=&
\argmin_{\tau }
\log \Big(\Tr \Big(\exp
(\log \rho_A - \frac{1}{\gamma}\Omega(\rho_{A}) + 
\sum_{j=1}^k\tau^j H_j )\Big)\Big) \notag\\
&-\sum_{j=1}^k\tau^j E_j.\Label{BNC}
\end{align}

\section{Numerical experiments}\label{sec:experiment}
To verify the effectiveness of our method, we investigate the performance of our method in calculating the quantum relative entropy of channel by conducting numerical experiments. 
We compare our method with the SDP-based approach in \cite{wilde2025sdp}, and we consider the same example as in \cite{wilde2025sdp}: the channel relative entropy between a dephasing (phase-flip) channel and a depolarizing channel.
The depolarizing channel $\cD_p$ with parameter $p$ is defined by
$ \cD_p(\rho) = (1-p)\rho + p\frac{I}{2}$, and
the dephasing channel is defined as 
$    \cD_{deph}(\rho) = p_{deph}\rho + (1-p_{deph})\sigma_z\rho\sigma_z$.
where $\sigma_z$ denotes the Pauli Z matrix. 
The parameter of dephasing channel $\cD_{deph}$ is fixed by $p_{deph} = 0.4$ and the depolarizing parameter varies in the range $p\in[0,0.1]$. In order to ensure the convergence of our algorithm, we need to choose a proper $\gamma$ such that it can satisfy the conditions in Theorem~\ref{thmA} and Theorem~\ref{thmB}. In our numerical experiments, we choose $\gamma=1$ and the initial point $\rho^{(1)}$ is randomly chosen from the $\cS(\cH_A)$.
Fig.~\ref{fig:exp_4} shows the relative entropy of channels $D(\cD_{deph}\|\cD_p)$ as a function of depolarizing parameter $p$. 
We can see that the numerical results we obtain are consistent with the results shown in \cite{wilde2025sdp}. 

Next, we numerically verify Conditions (a1), (a2), and (a3).
To assess Condition (a1), we compute the maximum and minimum of the ratio 
$D_{\Omega}(\rho^{(100)}\|\sigma)/D(\rho^{(100)}\|\sigma)$
over 10,000 randomly generated states $\sigma$ in a neighborhood of 
$\rho^{(100)}$, generated in the following way, as plotted in Fig. \ref{fig:exp_5}.
To choose $\sigma$ from the neighborhood of $\rho^{(100)}$,
we apply small perturbations on the point $\rho^{(100)}$.
In each trial, we draw a random Hermitian perturbation $H$ and set 
$\sigma \propto \rho^{(100)}+\epsilon H$, where $\epsilon\in(0,0.1]$ is chosen uniformly at random.
Since the point $\rho^{(100)}$ is close to the ideal convergent, 
Fig. \ref{fig:exp_5} shows that Condition (a1) holds.

To verify Condition (a3), we calculate
the maximum value of ratio $D_{\Omega}(\rho^{(j+1)}\|\rho^{(j)})/D(\rho^{(j+1)}\|\rho^{(j)})$ across sequential state pairs $j=1,..., 99$ for varying parameter $p$.
As shown in Fig.~\ref{fig:exp_7}, Condition (a3) holds for all tested values of $p$ 
except $p=0.052$ and $p=0.076$.
Indeed, this property depends on the choice of the initial state $\rho^{(1)}$.
It is expected that, if we choose another initial state, 
we still have a high possibility to satisfy Condition (a3) under the choices of 
$p=0.052, 0.076$.
To verify Condition (a2), we calculate
the maximum and minimum values of ratio 
$D_{\Omega}(\rho^{(100)}\|\rho^{(j)})/D(\rho^{(100)}\|\rho^{(j)})$ across sequential states $j=1,..., 99$ for varying parameter $p$,
as plotted in Fig. \ref{fig:exp_9}.
Since the value 
$D_{\Omega}(\rho_*\|\rho^{(j)})/D(\rho_*\|\rho^{(j)})$
is continuous for $\rho_*$ and 
$\rho^{(100)}$ is close to $\rho_*$,
Fig. \ref{fig:exp_9} suggests that 
$D_{\Omega}(\rho_*\|\rho^{(j)})/D(\rho_*\|\rho^{(j)})
\ge 0$, which implies Condition (a2). 
In fact, since the dephasing channel and depolarizing channel are both Pauli channels, these two channels are teleportation-covariant channels \cite{pirandola2017fundamental}. Thus according to the property of teleportation-covariant channel and data processing inequality, 
the maximization in (\ref{eq: def1 of qre of channel}) of these two channels is reached at the maximally entangled state, which yields the exact value of $D(\cD_{deph}\|\cD_p)$.
For completeness, we show the gap value between theoretical results and our numerical results in Fig.~\ref{fig:error}, which confirms 
that our obtained results are sufficiently close to the true values.
These evaluations enable us to apply the bound in \eqref{XME} except for $p=0.052, 0.076$.
In fact, these exceptional cases $p=0.052, 0.076$
show the impossibility to analytically show Condition (a2) with $\gamma=1$ 
in this example.

\begin{figure}[htbp]
    \centering
    \includegraphics[width = \columnwidth]{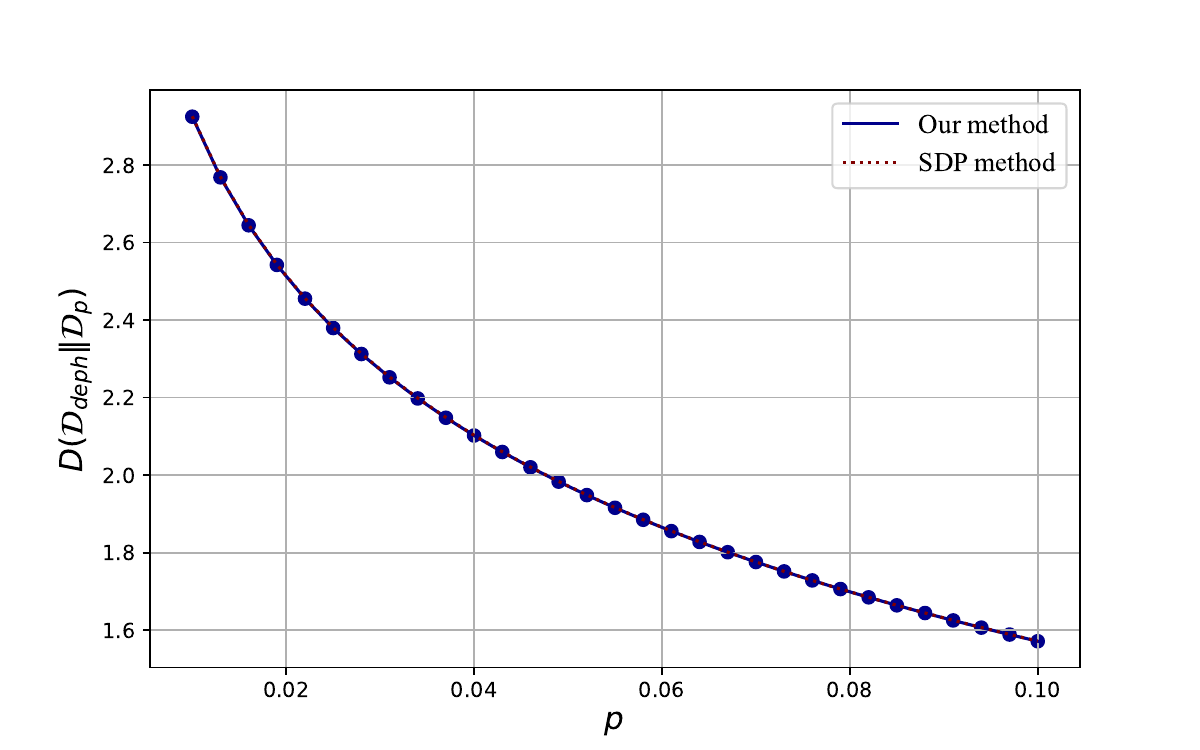}
    \caption{\textbf{Relative entropy of channels $D(\cD_{deph}\|\cD_p)$ as a function of the depolarizing parameter $p$.}
    The vertical axis shows the value of the $D(\cD_{deph}\|\cD_p)$.
    The horizontal axis shows the value of $p$. The dephasing parameter is fixed as 0.4. The blue line represents the numerical results of our method and the red dot line represents the results shown in \cite{wilde2025sdp}}
    \label{fig:exp_4}
\end{figure}

\begin{figure}[htbp]
    \centering
    \includegraphics[width = \columnwidth]{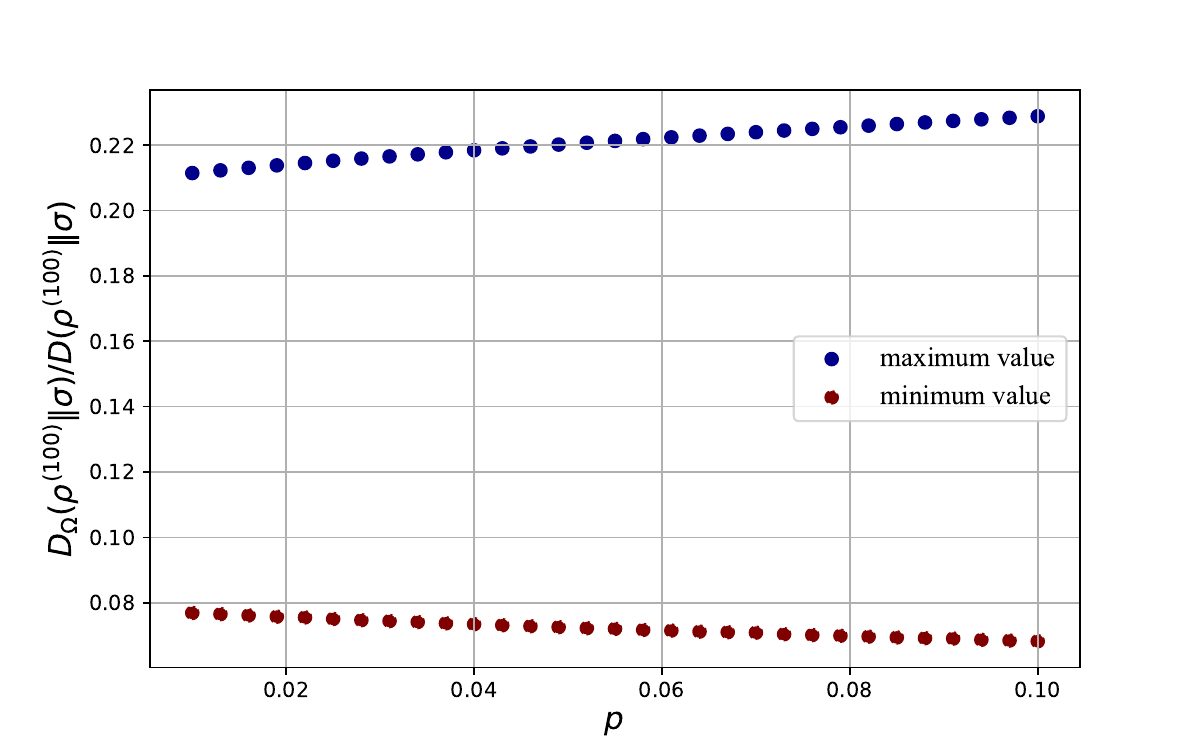}
    \caption{\textbf{The maximum value and minimum value of ratio $D_{\Omega}(\rho^{(100)}\|\sigma)/D(\rho^{(100)}\|\sigma)$ for the convergent points in Fig.~\ref{fig:exp_4}.}
    $\rho^{(100)}$ denotes the convergent points we get in Fig.~\ref{fig:exp_4}. Blue points denotes the maximum value of $D_{\Omega}(\rho^{(100)}\|\sigma)/D(\rho^{(100)}\|\sigma)$ among 10000 random states $\sigma$. Red points denotes the minimum value of $
    D_{\Omega}(\rho^{(100)}\|\sigma)/D(\rho^{(100)}\|\sigma)$ among 10000 random states $\sigma$}.
    \label{fig:exp_5}
\end{figure}

\begin{figure}[htbp]
    \centering
    \includegraphics[width = \columnwidth]{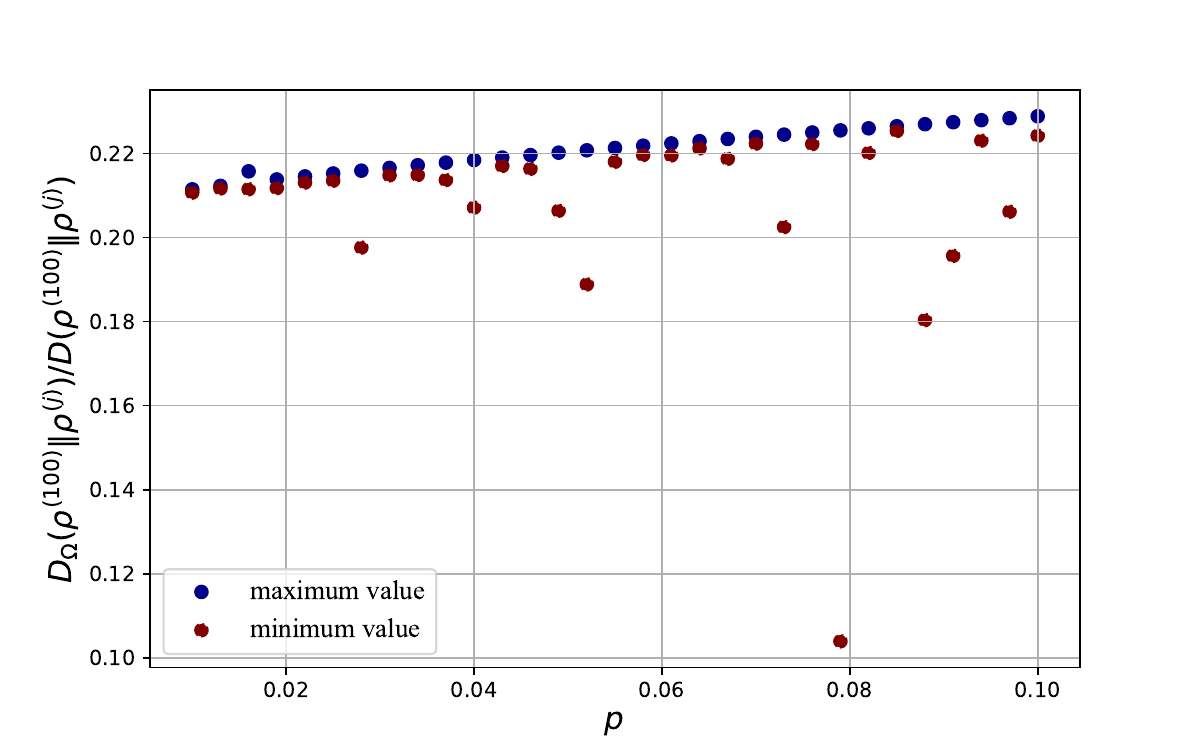}
    \caption{\textbf{The maximum value and minimum value of ratio $D_{\Omega}(\rho^{(100)}\|\rho^{(j)})/D(\rho^{(100)}\|\rho^{(j)})$ for the convergent points in Fig.~\ref{fig:exp_4}.}
    $\rho^{(100)}$ denotes the convergent points we get in Fig.~\ref{fig:exp_4}. Blue points denotes the maximum value of $D_{\Omega}(\rho^{(100)}\|\rho^{(j)})/D(\rho^{(100)}\|\rho^{(j)})$ for $j=1, \ldots, 99$. Red points denotes the minimum value of $
    D_{\Omega}(\rho^{(100)}\|\rho^{(j)})/D(\rho^{(100)}\|\rho^{(j)})$ for $j=1, \ldots, 99$.}
    \label{fig:exp_9}
\end{figure}

\begin{figure}[htbp]
    \centering
    \includegraphics[width = \columnwidth]{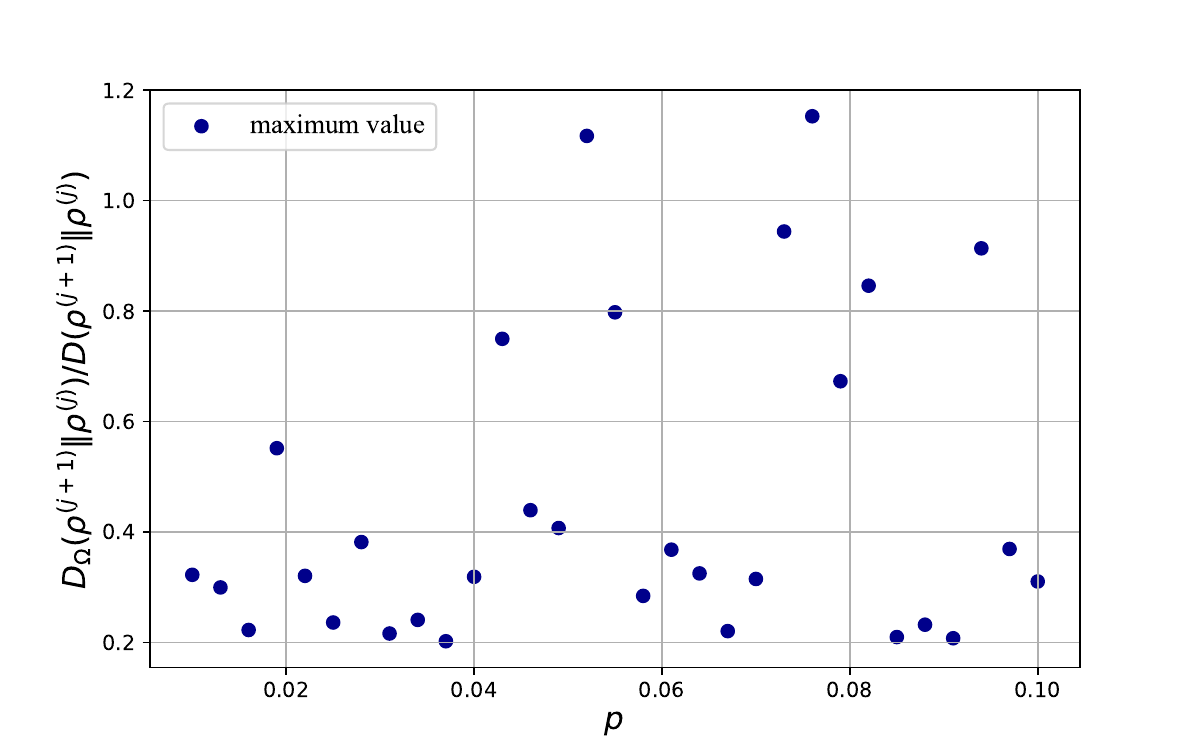}
    \caption{\textbf{The maximum value of ratio $D_{\Omega}(\rho^{(j+1)}\|\rho^{(j)})/D(\rho^{(j+1)}\|\rho^{(j)})$ across sequential state pairs ($j=1,..., 99$) for varying parameter $p$.}}.
    \label{fig:exp_7}
\end{figure}

\begin{figure}[htbp]
    \centering
    \includegraphics[width = 0.9\columnwidth]{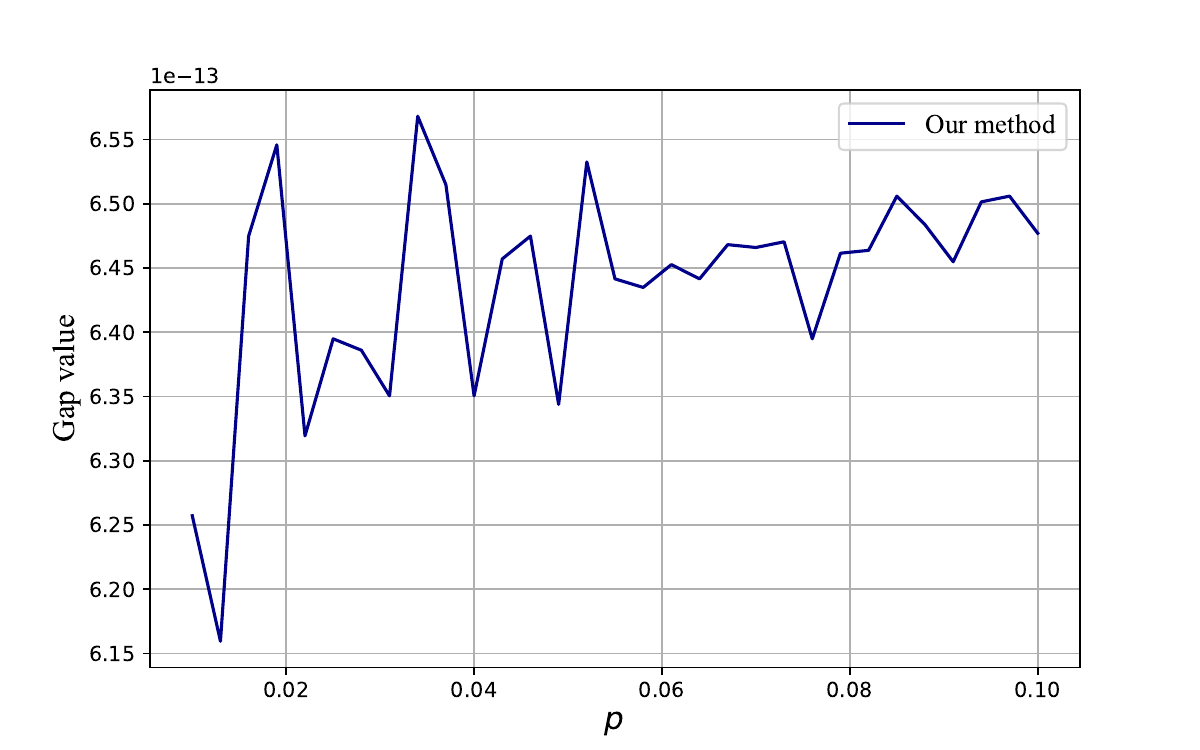}
    \caption{\textbf{The plot of the gap value between exact value and our numerical results considered in Fig.~\ref{fig:exp_4}.}
    The vertical axis shows the gap value.
    The horizontal axis shows the value of $p$.  }
    \label{fig:error}
\end{figure}

\begin{figure}[htbp]
    \centering
    \includegraphics[width = \columnwidth]{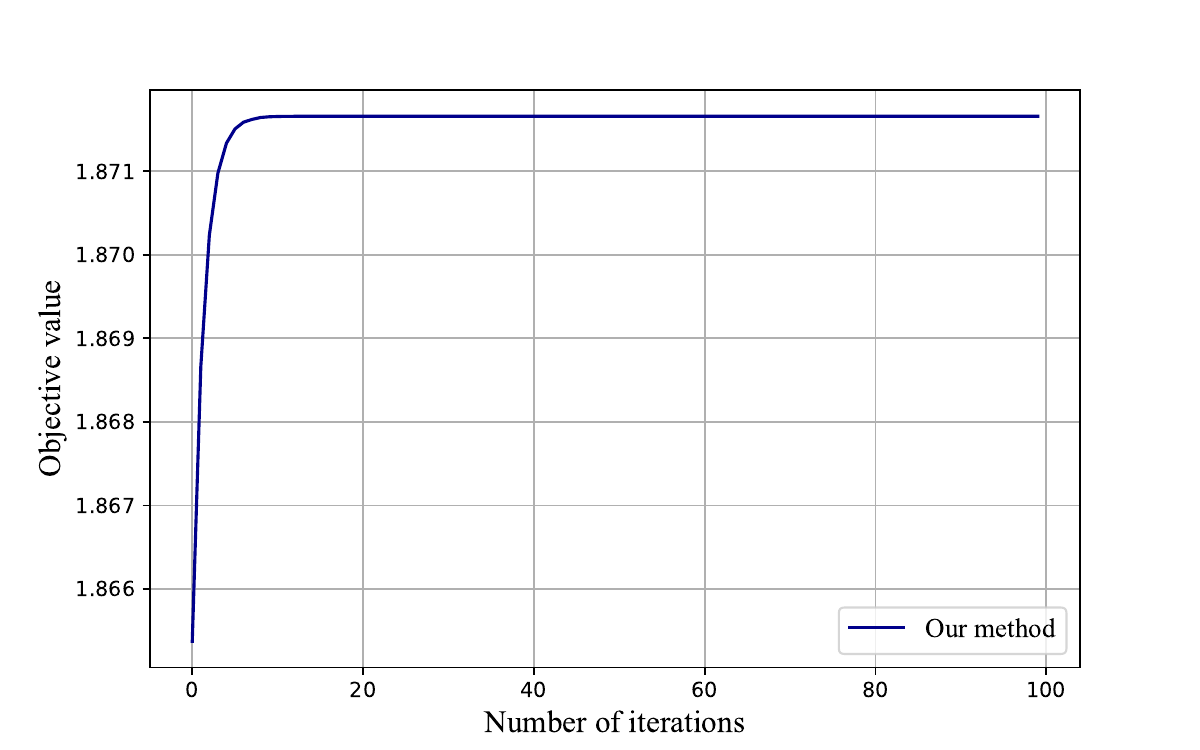}
    \caption{\textbf{Iterative behavior of our method with energy constraints.}
    The vertical axis shows the value of the objective function.
    The horizontal axis shows the number of iterations. The blue line represents the value of objective function in each iterations. The energy constraint is $\Tr\rho_A\sigma_z =-0.25$.}
    \label{fig:exp_6}
\end{figure}

Also, we present Fig.~\ref{fig:exp_6} to demonstrate the applicability of our method when we have energy constraints. In this case, we consider the relative entropy between  dephasing channel with parameter $p_{deph}=0.4$ and depolarizing channel with parameter $p=0.05$. The energy constraint is described by restricting the feasible region as the set $S_{(\sigma_z,-0.25)}:=\{\rho_A\in\cS(\cH_A): \Tr\rho_A\sigma_z = -0.25\}$. Notably, this figure shows our method remains effective in the case of energy-constrained systems.

\begin{table}[t]
\centering
\caption{Comparison between SDP-based methods and the quantum Arimoto--Blahut (QAB) algorithm for computing quantum relative entropy of channels.}
\label{tab:QAB_vs_SDP}
\resizebox{\columnwidth}{!}{%
\renewcommand{\arraystretch}{1.25}
\begin{tabular}{lcc}
\hline
\textbf{Aspect} & \textbf{SDP method in \cite{wilde2025sdp}} & \textbf{QAB algorithm} \\
\hline

Algorithmic type
& Convex optimization
& Fixed-point iteration\\

Solver required
& Yes
& No\\

Matrix variables 
& $O(\sqrt{\lambda/\varepsilon})$ 
& $1$ \\

Dimension of variables 
& $(d_A d_B)\times(d_A d_B)$ 
& $d_A\times d_A$ \\

Precision control 
& By SDP size ($\varepsilon$) 
& By iteration count \\

Memory scaling 
& Grows rapidly as $\varepsilon\to 0$ 
& Independent of $\varepsilon$ \\

\hline
\end{tabular}%
}
\end{table}

We then present Table.\ref{tab:QAB_vs_SDP} that illustrate the structural comparison between the SDP approaches in \cite{wilde2025sdp} and our quantum Arimoto-Blahut method for computing the quantum relative entropy of channels. In their SDP formulations, the optimization is performed over multiple free positive semidefinite matrix variables acting on the joint Hilbert space. The number of such matrix variables scales as $O(\sqrt{\lambda/\varepsilon})$ with the target precision, where $\lambda$ is a fixed value related to the max-relative entropy of the Choi matrices of two channels. In contrast, it is worth noting that the number of variables of the optimization problem shown in \eqref{BNB} is typically only dependent on the dimension of the Hilbert space. 

More specifically, when no linear constraint is imposed,
we do not need the $e$-projection procedure as discussed above
so that we do not need to solve the optimization problem \eqref{BNB} in this task. Thus, we only have one matrix-valued variable $\rho_A$ on the input Hilbert space $\cH_A$. Although the update rule involves operators on an enlarged Hilbert space, these operators are deterministic functions of $\rho_A$ and fixed problem data, and no optimization over matrices on the enlarged space is performed. 
Therefore, one of the main advantages of our method is that our method proceeds via a deterministic fixed-point iteration and does not require solving any semidefinite program or invoking a generic optimization solver 
while their method needs to solve a sequence of large scale convex optimization problem.

Then, for the above example calculating the channel relative entropy of dephasing channel and depolarizing channel, their method needs to optimize more than one hundred $4\times 4$ matrix to achieve the accuracy at least $\epsilon=10^{-2}$ with certain depolarizing parameters. And this number of matrix variables would increase exponentially when we consider multiple independent uses of channels. Therefore, with much fewer variables and lower dimension of variables, our method offers superior scalability and computational efficiency compared to their method . Moreover, remember that the precision of their method is closely related to the number of variables of SDP but the precision of our method is only dependent on the number of iterations. Therefore, the memory requirements of their SDP method grows rapidly as precision $\epsilon\to 0$ while the memory usage of our method is independent of the target precision. Besides this, while their method need substantial modifications and re-implementation for higher precision requirements, our approach maintains the same implementation across all accuracy levels - only tolerance parameter or the number of iterations need adjustment.

\section{Discussion}\label{sec:discussion}
In this paper, we have advanced the framework of the generalized quantum Arimoto–Blahut algorithm and demonstrated its efficacy in solving a challenging problem in quantum information theory.
First, we addressed a significant theoretical gap by proving the convergence of the algorithm to the global minimum under the condition of convexity combined with a certain natural condition. This result generalizes the applicability of the Arimoto–Blahut framework, moving beyond the limitations of previously known convergence criteria that were often analytically intractable for novel quantities.

Second, we applied this refined algorithm to the computation of the quantum relative entropy of channels. This task has been notoriously difficult because standard convex optimization methods, such as gradient descent, are impractical due to the extreme complexity of computing the analytical derivatives. While a recent alternative approach \cite{wilde2025sdp} employs semidefinite programming (SDP), it requires handling an enormous number of matrix variables, leading to prohibitive computational costs. Our method overcomes these issues by providing a computationally efficient alternative. By leveraging our first-step theoretical result, we guaranteed that the algorithm converges to the global optimum, successfully transforming an analytically and computationally daunting problem into a manageable iterative process.

Our numerical results for qubit channels illustrate several practical advantage of the proposed method. Compared with recent SDP-based approaches, out method requires only a single matrix variable on the input space and avoids introducing a large number of auxiliary matrix variables. As a result, the memory cost of our methods is independent of the target precision, in contrast to SDP formulations whose resource requirements grow rapidly as higher accuracy is demanded. Moreover, precision control in our approach is achieved simply by adjusting the iteration count or stopping tolerance, without any structural modification of the implementation. These advantages make our method notably more scalable and adaptable to higher-dimensional systems.

Finally, we note that the numerical experiments also highlight the necessity of \emph{a posteriori} verification.
In particular, for $p=0.052$ and $p=0.076$ we observed that the certification condition (e.g., (a3) with $\gamma=1$) may fail for certain initializations, illustrating that an \emph{a priori} analytical verification of the condition can be difficult in our concrete instances; this is precisely where our \emph{a posteriori} viewpoint becomes essential (Sec.~\ref{sec:experiment}).

Nevertheless, there are three remaining open problems and directions for future work. First, while we focused on the quantum relative entropy of channels, the methodology presented here is not limited to this specific quantity. The numerical verification paradigm can be applied to a wide class of optimization problems in quantum information theory. Thus it is an interesting topic to find its other applications. Second, even though our numerical verification method is empirically reliable, a more rigorous statistical or adaptive sampling strategy could further enhance its confidence guarantees. Third, since the choice of $\gamma$ influences convergence speed and stability, an adaptive $\gamma$-selection rule may improve the performance of the algorithm.

\appendices

\section{Proof of \eqref{ALA1}}
The definition of $D(\cN\|\cM)$
implies 
\begin{align} 
    &-D(\cN\|\cM) \notag\\
    =& \inf_{\rho_A \in \cS(\cH_A)} -D(\sqrt{\rho_A} \Gamma_{AB}^{\cN}\sqrt{\rho_A} \|\sqrt{\rho_A}\Gamma_{AB}^{\cM}\sqrt{\rho_A}) \notag\\
    =& \inf_{\rho_A \in \cS(\cH_A)} -\tr(\rho_A^{\frac{1}{2}} \Gamma_{AB}^{\cN}\rho_A^{\frac{1}{2}} (\log \rho_A^{\frac{1}{2}} \Gamma_{AB}^{\cN}\rho_A^{\frac{1}{2}} - \log \rho_A^{\frac{1}{2}} \Gamma_{AB}^{\cM}\rho_A^{\frac{1}{2}})) \notag\\
    =& \inf_{\rho_A \in \cS(\cH_A)} -\tr(\rho_A^{-\frac{1}{2}}\rho_A^{\frac{1}{2}}\rho_A^{\frac{1}{2}}\Gamma_{AB}^{\cN}\rho_A^{\frac{1}{2}} \notag\\
&\cdot    (\log \rho_A^{\frac{1}{2}} \Gamma_{AB}^{\cN}\rho_A^{\frac{1}{2}} - \log \rho_A^{\frac{1}{2}} \Gamma_{AB}^{\cM}\rho_A^{\frac{1}{2}})) \notag\\
    =& \inf_{\rho_A \in \cS(\cH_A)} -\tr(\rho_A \Gamma_{AB}^{\cN}\rho_A^{\frac{1}{2}} \notag\\
&\cdot    (\log \rho_A^{\frac{1}{2}} \Gamma_{AB}^{\cN}\rho_A^{\frac{1}{2}} - \log \rho_A^{\frac{1}{2}} \Gamma_{AB}^{\cM}\rho_A^{\frac{1}{2}})\rho_A^{-\frac{1}{2}}) \notag\\
    =& \inf_{\rho_A \in \cS(\cH_A)} -\tr(\rho_A \tr_B(\Gamma_{AB}^{\cN}\rho_A^{\frac{1}{2}} (\log \rho_A^{\frac{1}{2}} \Gamma_{AB}^{\cN}\rho_A^{\frac{1}{2}} \notag\\
    &  - \log \rho_A^{\frac{1}{2}} \Gamma_{AB}^{\cM}\rho_A^{\frac{1}{2}})\rho_A^{-\frac{1}{2}}))
\label{optimization problem of qre}.
\end{align}
 Since $\tr[\rho_{A}\Omega_1(\rho_{A})] \in \mathbb{R}$, we have
    \begin{align}        \tr[\rho_{A}\Omega_1(\rho_{A})] = \tr[\Omega_1^{\dagger}(\rho_{A})\rho^*_{A}] = \tr[\rho_{A}\Omega_1^{\dagger}(\rho_{A})] .
    \end{align}
Thus, we have
$    \cG(\rho_{A}) = \tr[\rho_{A}\Omega(\rho_{A})]$.

\section{Proof of Theorem \ref{thmA}}
To show Theorem \ref{thmA}, we prepare the following lemma.
\begin{lemma}\label{lemma:BB}
We consider a parameterized state $\rho(\theta)$ with $\theta \in \mathbb{R}^d$. 
We assume that each matrix component of $\Omega[\rho]$
is differentiable function of the matrix components of $\rho$.
In addition, we assume that the relation \eqref{BK1+} holds
when $\rho = \rho(\theta_0)$ and 
$\sigma$ is an arbitrary element of the neighborhood of $\rho(\theta_0)$.
Then, the relation
\begin{align}
\Tr \rho(\theta) (\frac{\partial }{\partial \theta^j}\Omega[\rho(\theta)])\Big|_{\theta=\theta_0}=0 \label{VB1},
\end{align}
i.e., 
\begin{align}
\frac{\partial }{\partial \theta^j} (\Tr \rho(\theta) \Omega[\rho(\theta)])
\Big|_{\theta=\theta_0}=
 \Tr 
\Big( \frac{\partial }{\partial \theta^j}\rho(\theta) 
  \Big|_{\theta=\theta_0}\Big)
\Omega[\rho(\theta_0)])
\label{BXN}
\end{align}
holds for $j=1, \ldots, d$.
\end{lemma}

\begin{proof}
Since it is sufficient to show \eqref{VB1} in the case with $d=1$,
we discuss this case.
The condition \eqref{BK1+} implies
\begin{align}
\Tr \rho(\theta) (\Omega[\rho(\theta)]-\Omega[\rho(\theta\pm\epsilon)])
\le \gamma D(\rho(\theta)\|\rho(\theta\pm\epsilon)) \label{BP1} 
\end{align}
for $\epsilon >0$.
When $\epsilon $ is close to $0$, 
$D(\rho(\theta)\|\rho(\theta\pm \epsilon))$ behaves as $O(\epsilon^2)$.
$\Tr \rho(\theta) (\Omega[\rho(\theta)]-\Omega[\rho(\theta\pm \epsilon)])
$ is characterized as
$\pm \Tr \rho(\theta) (\frac{d}{d\theta}\Omega[\rho(\theta)]) \epsilon 
+O(\epsilon^2)$.
Thus, \eqref{BP1} implies
\begin{align}
\pm \Tr \rho(\theta) (\frac{d}{d\theta}\Omega[\rho(\theta)]) \le 0,
\end{align}
which implies \eqref{VB1}.
\end{proof}

Now, we start the proof of Theorem \ref{thmA}.
The point $\theta_0$ realizes the minimum value if and only if 
\begin{align}
\frac{\partial }{\partial \theta^j} (\Tr \rho(\theta) \Omega[\rho(\theta)])\Big|_{\theta=\theta_0}=0.
\label{BXN2}
\end{align}
Lemma \ref{lemma:BB} guarantees that 
the condition \eqref{BXN2} 
is equivalent to the condition (b2).
Thus, the condition (b1) is equivalent to the condition (b2).

Next, we assume the condition (b3).
When $\rho_0$ is a convergent of Algorithm \ref{AL1},
we have 
\begin{align}
\Gamma^{(e)}_{{\cal M}}[{\cal F}_3[\rho_0]]=\rho_0,
\end{align}
which implies that 
there exists parameters $\tau_{j,0}$ such that 
\begin{align}
\log \rho_0 -\frac{1}{\gamma}\Omega[\rho_0]= \log \rho_0+\sum_{j=0}^k H_j \tau_{j,0}
\end{align}
Thus, we have
\begin{align}
\Omega[\rho_0]= -\gamma \sum_{j=0}^k H_j \tau_{j,0}.
\end{align}
When $\rho(\theta_0)$ is a convergent $\rho_0$ of Algorithm \ref{AL1},
we have
\begin{align}
&\Tr \frac{\partial }{\partial \theta^j}\rho(\theta)|_{\theta=\theta_0} 
\Omega[\rho(\theta_0)] \notag \\
= &
-\Tr \frac{\partial }{\partial \theta^j}\rho(\theta)|_{\theta=\theta_0} 
\gamma \sum_{j=0}^k H_j \tau_{j,0}
=0,\label{BNT}
\end{align}
which implies the condition (b2).
Thus, the condition (b3) implies the conditions (b2) and (b1).
Due to the uniqueness of the minimizer of 
$\min_{\rho \in {\cal M}} \Tr \rho \Omega[\rho]$, 
the minimizer is limited to the convergent of Algorithm \ref{AL1}.
Hence, the condition (b1) implies the condition (b3).

\section{Proof of Theorem \ref{thmB}}
First, to show 
\begin{align}
{\cal G}(\rho^{(j)})
\ge {\cal G}(\rho^{(j+1)}),\label{SK0}
\end{align}
we introduce the function $J(\rho,\sigma)$ as
\begin{align}
J(\rho,\sigma):=\gamma D(\rho\|\sigma)
+\Tr \rho \Omega(\sigma).
\end{align}
The reference \cite[Lemma 3]{hayashi2024ab} shows that
\begin{align}
{\cal G}(\rho^{(j)})=J(\rho^{(j)},\rho^{(j)})
\ge J(\rho^{(j+1)},\rho^{(j)}). \label{SK1}
\end{align}
Also, we have
\begin{align}
&J(\rho^{(j+1)},\rho^{(j)})- {\cal G}(\rho^{(j+1)}) \notag\\
=&\gamma D(\rho^{(j+1)}\|\rho^{(j)})
-D_\Omega(\rho^{(j+1)}\|\rho^{(j)})
\stackrel{(a)}{\ge} & 0,\label{SK2}
\end{align}
where
$(a)$ follows from the assumption of Theorem \ref{thmB}.
The combination of \eqref{SK1} and \eqref{SK2}
implies \eqref{SK0}.

Next, we notice the following relation.
\begin{align}
&\gamma(D(\rho_*\| \rho^{(j)})- D(\rho_*\| \rho^{(j+1)})) \notag\\
&-({\cal G}(\rho^{(j+1)})-{\cal G}(\rho_*))\notag\\
\stackrel{(b)}{=} &\gamma D(\rho^{(j+1)}\|\rho^{(j)})
-D_\Omega(\rho^{(j+1)}\|\rho^{(j)})+D_\Omega(\rho_*\|\rho^{(j)}) \notag\\
\stackrel{(c)}{\ge}  &0, \label{SK5}
\end{align}
where $(b)$ follows from 
the reference \cite[Lemma 5]{hayashi2024ab}, 
and $(c)$ follows from the assumption of Theorem \ref{thmB}.

Thus, we have
\begin{align}
&\gamma(D(\rho_*\| \rho^{(1)})
\ge \gamma(D(\rho_*\| \rho^{(1)})- D(\rho_*\| \rho^{(t)})) \notag\\
=&
\sum_{j=1}^{t-1}
\gamma(D(\rho_*\| \rho^{(j)})- D(\rho_*\| \rho^{(j+1)})) \notag\\
\stackrel{(d)}{\ge} &
\sum_{j=1}^{t-1}
{\cal G}(\rho^{(j+1)})-{\cal G}(\rho_*)
\stackrel{(e)}{\ge}
\sum_{j=1}^{t-1}
{\cal G}(\rho^{(t)})-{\cal G}(\rho_*) \notag\\
=&
(t-1)
({\cal G}(\rho^{(t)})-{\cal G}(\rho_*)),
\end{align}
which implies \eqref{XME},
where 
$(d)$ follows from \eqref{SK5} and
$(e)$ follows from \eqref{SK0}.

\bibliography{references}
\bibliographystyle{plain}
\end{document}